\documentclass[10pt,aps,pra,twocolumn,showpacs,superscriptaddress]{revtex4-2}
\usepackage{amsthm}
\usepackage{amssymb}
\usepackage{bm}
\newtheorem{theorem}{Theorem} 
\newtheorem{proof*}{Proof}
\theoremstyle{remark}
\newtheorem{remark}{Remark}
\usepackage{graphicx}
\usepackage{booktabs}
\usepackage[colorlinks=true,linkcolor=blue,citecolor=blue,urlcolor=magenta]{hyperref}
\usepackage{amsmath}

\begin{document}  
\title {\bf Advancing Quantum Process Tomography through Quantum Compilation}

\author{Huynh Le Dan Linh} 
\affiliation{University of Information Technology, Ho Chi Minh City, 700000, Vietnam}
\affiliation{Vietnam National University, Ho Chi Minh City, 700000, Vietnam}

\author{Vu Tuan Hai}
\affiliation{University of Information Technology, Ho Chi Minh City, 700000, Vietnam}
\affiliation{Vietnam National University, Ho Chi Minh City, 700000, Vietnam}
\affiliation{Nara Institute of Science and Technology, 8916–5 Takayama-cho, Ikoma, Nara 630-0192, Japan}

\author{Le Bin Ho} 
\thanks{Electronic address: binho@fris.tohoku.ac.jp}
\affiliation{Department of Applied Physics, 
Graduate School of Engineering, 
Tohoku University, 
Sendai 980-8579, Japan}
\affiliation{Frontier Research Institute 
for Interdisciplinary Sciences, 
Tohoku University, Sendai 980-8578, Japan}

\date{\today}

\begin{abstract}
Quantum process tomography (QPT) plays a central role in characterizing quantum gates and circuits, diagnosing quantum devices, calibrating hardware, and supporting quantum error correction. However, conventional QPT methods face challenges related to scalability and sensitivity to noise. In this work, we propose a QPT framework based on quantum compilation, which represents quantum processes using optimized Kraus operators and Choi matrices. By formulating QPT as a compilation and optimization problem, our approach significantly reducing measurement and computational overhead while maintaining reconstruction accuracy. We benchmark the method using numerical simulations of Haar-random unitary gates and demonstrate a reliable process reconstruction. We further apply the framework to dephasing channels with both time-homogeneous and time-inhomogeneous noise, as well as to depolarizing and amplitude-damping channels, where stable performance is observed across different noise regimes. These results indicate that quantum compilation–based QPT can serve as a practical alternative to standard QPT methods for quantum process characterization and device validation.
\end{abstract}
\maketitle

\section{Introduction}
Quantum process tomography (QPT) is a standard technique or characterizing unknown quantum processes \cite{Chuang01111997,PhysRevLett.78.390,nielsen2000quantum}. A standard approach involves preparing a set of quantum states, applying the target process, and reconstructing the output states via quantum state tomography \cite{Chuang01111997}. The QPT is crucial for verifying and characterizing the operation of quantum gates and circuits \cite{PhysRevLett.78.390,PhysRevLett.93.080502,PhysRevLett.97.220407,Tinkey_2021,Bialczak2010}.
However, the conventional QPT suffers from scalability limitations: both measurement and computational costs grow exponentially with system size, making it infeasible for large quantum systems \cite{9992920,10386455,PhysRevLett.106.100401,PhysRevA.108.032419}.
As quantum devices continue to scale, developing more efficient QPT methods that maintain high precision while reducing resource requirements has become increasingly important.

Recent advancements have aimed to improve the efficiency of QPT. For example, Ahmed et al. \cite{PhysRevLett.130.150402} proposed a machine learning-based approach that uses the Riemannian gradient descent to optimize a set of Kraus operators using, allowing efficient characterization of quantum processes through iterative adjustments. Building on this, Daniel et al. \cite{volya2024fastquantumprocesstomography} showed that the Riemannian gradient descent accelerates the optimization of quantum processes by exploiting the structure of the parameter space through the Riemannian geometry \cite{Jiang2015,PhysRevA.107.062421}. These approaches demonstrate the potential of machine learning and advanced optimization techniques to improve the QPT efficiency \cite{Torlai2023,PhysRevA.105.032427,PhysRevA.108.032419}. However, despite these improvements, these methods remain computationally expensive due to the large number of measurements required, since the tomography step still relies on the standard QPT methods.

In this work, we introduce a quantum compilation-based QPT (CQPT) framework, which combines quantum compilation techniques\cite{Hai2023,TuanHai2024} with QPT methods to optimize the tomography process, reduce computational costs, and maintain high precision.
A standard quantum compilation converts given target unitaries into trainable unitaries, enabling applications such as gate optimization \cite{heya2018variationalquantumgateoptimization}, quantum-assisted compiling \cite{Khatri2019quantumassisted}, continuous-variable targets \cite{PRXQuantum.2.040327}, quantum state tomography \cite{Hai2023}, and quantum dynamics simulation \cite{HAI2024101726}. In the CQPT, we extend the target unitaries to quantum processes, and the trainable unitaries become the trainable Kraus operators or Choi matrices.

The core innovation of this framework is the application of quantum compilation, a quantum machine learning technique, to optimize QPT, thereby reducing the complexity of the traditional QPT methods.
This approach offers a scalable framework that can enable more practical QPT than existing methods.

To achieve this, we develop two key theorems. The first theorem addresses the optimization of Kraus operators through the quantum compilation, while the second extends a similar framework for Choi matrices. The Kraus operators representation is crucial for characterizing unitary processes, as it provides a direct description of the quantum channel in terms of experimentally determinable operators. Similarly, Choi matrices provide a compact representation of a quantum operation, which is particularly valuable for understanding the behavior of quantum channels, especially in noisy environments such as open quantum systems.
We demonstrate the practical applications of these theorems in several quantum processes, including random unitary processes, dephasing processes, and quantum systems affected by both time-homogeneous and time-inhomogeneous noise. 

To optimize in the CQPT, we employ the Riemannian gradient descent method. This approach reduces the computational cost and incorporates the geometry of the quantum process space, resulting in a more accurate and efficient tomography process.

Through numerical simulations and analytical derivations, we demonstrate that the CQPT method significantly improves efficiency and reduces the computational complexity of the QPT, making it practical for large-scale quantum systems. This approach holds promise for both fundamental quantum information research and practical applications, including quantum computing and communication, where precise process characterization is essential.

\begin{figure*}[t]
    \centering
    \includegraphics[width=0.7\linewidth]{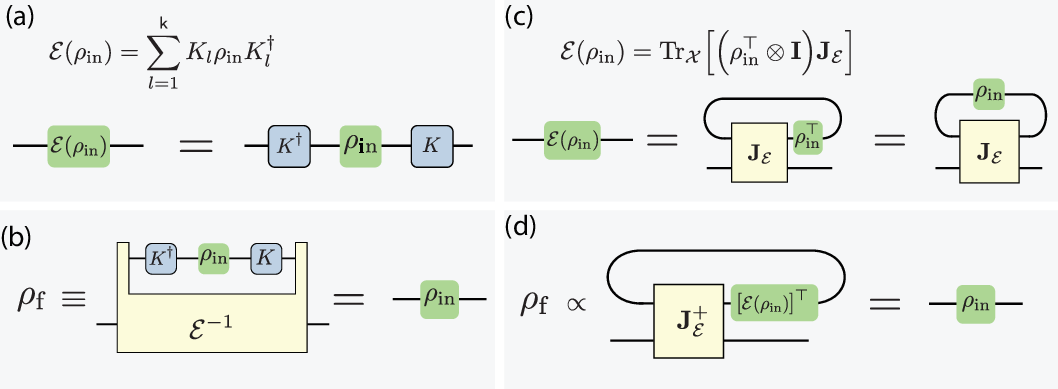}
    \caption{\textbf{Graphical presentation}.
    (a) Quantum channel representation with Kraus operators. (b) Visualization of Theorem \ref{theo:1}. (c) Quantum channel representation using the Choi matrix. (d) Visualization of Theorem \ref{theo:2}.
    }
    \label{fig:1}
\end{figure*}

\section{Compilation-based quantum process tomography framework}
QPT is a method for fully characterizing an unknown quantum process $\mathcal{E}$, which is modeled as a completely positive and trace-preserving (CPTP) map. The process $\mathcal{E}$ can be expressed in the operator-sum representation as
\(
\mathcal{E}(\rho) = \sum_l K_l \rho K_l^\dagger,
\)
where $K_l$ are Kraus operators satisfying the completeness relation $\sum_l K_l^\dagger K_l = \mathbf{I}$. Alternatively, the QPT represents $\mathcal{E}$ through its Choi matrix $\mathbf{J}_{\mathcal{E}}$, which is an operator spanning on $H_{\mathcal{X}} \otimes H_{\mathcal{Y}}$, ie., 
\(
\mathcal{E}(\rho) = {\rm Tr}_{\mathcal{X}}
    \big[\big(
        \rho^\top\otimes \mathbf{I}
    \big)\mathbf{J}_{\mathcal{E}}
    \big].
\)   
This matrix is constructed using the Choi-Jamio{\l}kowski isomorphism as
\(
\mathbf{J}_{\mathcal{E}} = \sum_{i,j} |i\rangle \langle j| \otimes \mathcal{E}(|i\rangle \langle j|)
\)
where $\{|i\rangle\}$ is an orthonormal basis. 
For an \(N\)-qubit system, the Hilbert space dimension is \(2^N\), with \(K_l \in \mathbb{C}^{2^N \times 2^N}\) and \(\mathbf{J}_\mathcal{E} \in \mathbb{C}^{4^N \times 4^N}\).

Experimentally, QPT proceeds by preparing a set of input states $\{\rho_i\}$, applying the quantum process $\mathcal{E}$, and measuring the output states using a positive operator-valued measure (POVM) $\{M_j\}$. The collected measurement data are used to extract the Kraus operators or to reconstruct the Choi matrix \( \mathbf{J}_{\mathcal{E}} \) via direct inversion methods \cite{Chuang01111997,PhysRevLett.78.390}, statistical estimation approaches \cite{Bouchard2019quantumprocess,Granade_2017,PhysRevLett.102.020504,PhysRevB.90.144504,Torlai2023}, or optimization-based techniques \cite{Torlai2023,PhysRevA.105.032427,PhysRevA.108.032419}. The QPT is essential for validating quantum gates \cite{Bialczak2010, Gaikwad2022}, benchmarking devices \cite{Bialczak2010}, and analyzing noise in quantum systems \cite{Bialczak2010}.

In this work, we present a compilation-based approach for quantum process tomography (CQPT). In the standard compilation, given a target unitary \( U \), a trainable unitary \( V(\bm\theta) \) is introduced and optimized over parameters \( \bm\theta \) such that \( V^\dagger(\bm\theta) U = \mathbf{I} e^{i\phi} \) for an arbitrary global phase \( \phi \). Here, we extend \( U \) to a CPTP process \( \mathcal{E} \), and generalize \( V(\bm\theta) \) to a set of $\mathsf{k}$ Kraus operators \( \Bbbk = (K_1, K_2, \dots, K_\mathsf{k})^\top \), $\Bbbk\in\mathbb{C}^{\mathsf{k}2^N\times 2^N}$, where \( \Bbbk^\dagger \Bbbk = \mathbf{I} \) and \( K_l \) are optimizable Kraus operators. The CQPT framework is then defined as follows.

\begin{theorem}\label{theo:1}
(Kraus-based CQPT):
A set of trainable Kraus operators $\Bbbk = \big(K_1,\cdots,K_\mathsf{k}\big)^\top,\ \Bbbk^\dagger\Bbbk = \mathbf{I}$, represents a quantum process $\mathcal{E}$, which is a fully invertible CPTP map, if the following condition holds
\begin{equation}
    \rho_{\rm f} \equiv \mathcal{E}^{-1}\Big(\sum_{l=1}^\mathsf{k} K_l \rho_{\rm in} K_l^\dagger\Big) = \rho_{\rm in},
\end{equation}
where $\rho_{\rm in(f)}$ is the initial (final) state, and $\mathcal{E}^{-1}$ is the inverse process of $\mathcal{E}$. Here, $\mathcal{E}^{-1}$ is a CPTP only when  
$\mathcal{E}$ is unitary and is generally non-CPTP for non-unitary processes.
 \end{theorem}

\begin{proof}
For an initial quantum state $\rho_{\rm in}$, its evolution under the process $\mathcal{E}$ is given by $
\rho_{\mathcal{E}} = \mathcal{E}(\rho_{\rm in})$. Alternatively, the evolved state can be written in terms of Kraus operators as $\rho_{\Bbbk} = \sum_l K_l \rho_{\rm in} K_l^\dagger.$
The set of Kraus operators $\Bbbk$ represents $\mathcal{E}$ if
$
\rho_{\mathcal{E}} = \rho_{\Bbbk},
$
which implies
$
\sum_l K_l \rho_{\rm in} K_l^\dagger = \mathcal{E}(\rho_{\rm in}).
$
Applying the inverse process $\mathcal{E}^{-1}$ to both sides gives
$
\rho_{\rm f} \equiv \mathcal{E}^{-1}\Big(\sum_l K_l \rho_{\rm in} K_l^\dagger\Big) = \rho_{\rm in}.
$
Thus, $\Bbbk$ is a valid representation of $\mathcal{E}$.      
\end{proof} 

\begin{remark} (Unitary and near-unitary processes).
We emphasize that Theorem~\ref{theo:1} applies to unitary dynamics and near-unitary processes.
    If $\mathcal{E}$ is a unitary process, i.e., 
    $\mathcal{E}(\rho) = U\rho U^\dagger$ for some unitaries $U$,
    then $\mathcal{E}^{-1}(\rho) = U^\dagger\rho U$.
    This approach is crucial for validating quantum gates and benchmarking quantum devices.  
\end{remark}

Figure~\ref{fig:1} (a) illustrates a quantum channel where the initial state \( \rho_{\rm in} \) is transformed by a set of Kraus operators. Figure~\ref{fig:1} (b) visualizes Theorem \ref{theo:1}, showing that, after applying the Kraus operators, the inverse operator \( \mathcal{E}^{-1} \) is performed. The final state \( \rho_{\rm f} \) returns to \( \rho_{\rm in} \) when \( \sum_l K_l \rho_{\rm in} K_l^\dagger = \mathcal{E}(\rho_{\rm in}) \).

The framework in Refs. \cite{volya2024fastquantumprocesstomography,PhysRevLett.130.150402} uses two evolution paths: one to obtain $\rho_{\mathcal{E}} = \mathcal{E}(\rho_{\rm in})$, and the other to obtain $\rho_{\Bbbk} = \sum_l K_l \rho_{\rm in} K_l^\dagger$. The Kraus operators $\Bbbk$ are then optimized by minimizing the cost function
\(
C(\Bbbk) = \sum_{ij} \Big( {\rm Tr} \big[ M_j (\rho^{i}_{\mathcal{E}} - \rho^{i}_{\Bbbk}) \big] \Big)^2,
\)
where \( i \in [1, \mathsf{n}] \) labels \( \mathsf{n} \) different input states \( \rho_{\rm in}^1, \dots, \rho_{\rm in}^\mathsf{n} \), and \( j \in [1, \mathsf{p}] \) corresponds to \( \mathsf{p} \) elements of the POVM.  This method is resource-intensive because it requires full measurements of the POVM \( \{M_j\} \).

In our CQPT framework, we employ a quantum compilation process that benefits from ``single-shot'' measurements \cite{PhysRevLett.126.170504,Hai2023}. Starting from an initial pure state $\left| \psi \right\rangle$, we apply the trainable Kraus operators $\Bbbk$ followed by $\mathcal{E}^{-1}$, which yields the final state
$\rho_{\rm f} = \mathcal{E}^{-1}\Big(\sum_l K_l |\psi\rangle\langle\psi|K_l^\dagger\Big)$.
The cost function $C(\Bbbk)$ is defined through the infidelity as
\begin{eqnarray}\label{eq:ck_our}
C(\Bbbk) &=& 1 - \int_{\psi} \left| \mathrm{Tr}\left[\rho_{\mathrm{f}} |\psi\rangle\langle\psi| \right] \right|^2 \nonumber \\
        &=& 1 - \int_{W \in \mathrm{Haar}} \left| \mathrm{Tr} \left[ W^\dagger \rho_{\mathrm{f}} W |\mathbf{0}\rangle\langle \mathbf{0}| \right] \right|^2.
\label{eq:cost}
\end{eqnarray}
where we used $|\psi\rangle = W|\bm 0\rangle$ with an arbitrary 
random unitary $W$ and $|\bm 0\rangle \equiv |00\cdots0\rangle$ is the initial state of the quantum register.

Practically, in the quantum register, we first prepare $N$ qubits initialized by $|00\cdots 0\rangle$, apply $W$, followed by $\Bbbk$ and $\mathcal{E}^{-1}$, and apply $W^\dagger$ again. Then, we measure the final state in the basis \( |\bm{0}\rangle \) and obtain the probability \( p(00\ldots0) = {\rm Tr}\big[W^\dagger\rho_{\rm f}W|\bm{0}\rangle\langle\bm{0}|\big] \), which is used for optimization.
This method is known as single-shot measurement, where we only need the probability of the classical outcome $00\ldots0$, see also Refs.\cite{PhysRevLett.126.170504, Hai2023}.
In practice, $W\in{\rm Haar}$ is replaced by a finite set $W = \big\{W_i\,, \forall i \in [1, \mathsf{n}] \big\}$,
corresponding to $\mathsf{n}$ different input states.
The model is trained by updating \( \Bbbk \) iteratively until \( C(\Bbbk) \approx 0 \), ensuring that \( \Bbbk \) optimally represents the Kraus operators for the quantum channel \( \mathcal{E} \).

To optimize the cost function, we employ the Riemannian optimization method \cite{volya2024fastquantumprocesstomography, PhysRevLett.130.150402, PhysRevA.107.062421} and perform the gradient descent on the cost function $C(\Bbbk)$, which is defined on a Riemannian manifold. The Riemannian gradient is derived by projecting the Euclidean gradient onto the tangent space as
\begin{equation}
\rm {grad}\, C(\Bbbk) = \nabla C(\Bbbk) - \Bbbk \, \rm{Sym}(\Bbbk^\dagger \nabla C(\Bbbk)).
\end{equation}
where $ \rm {Sym}(\mathbf{A}) = (\mathbf{A} + \mathbf{A}^\dagger)/2 $. 
Here, we use `grad' to denote the Riemannian gradient and `$\nabla$' for the Euclidean gradient.
Using the Riemannian gradient descent, we update $ \Bbbk $ iteratively as
\begin{equation}\label{eq:optimizer}
    \Bbbk_{t+1} = \rm {Retract}_\Bbbk\big[-\alpha \, \rm {grad} C(\Bbbk)\big],
\end{equation}
where $\alpha$ is the learning rate, and the retraction operation ensures the update remains on the manifold. We employ an iterative Cayley transformation for the retraction, which approximates a second-order mapping of a tangent vector onto the manifold \cite{Tagare2011NotesOO}. Further details are provided in Appendix C.

Theorem~\ref{theo:1} also applies to near-unitary quantum processes.
Consider a weakly noisy channel with noise strength \(0 < \epsilon\ll 1\), described as
\(
\mathcal E(\rho) = 
(1-\epsilon)U\rho U^\dagger
+
\epsilon\frac{I}{2}.
\)
In this regime, an effective inverse can be defined at the linear order as
\(
\mathcal E^{-1}(\rho)
\approx
\frac{1}{1-\epsilon}U^\dagger\rho U,
\) where higher-order noise contributions have been neglected.

\begin{remark} (Irreversible processes). Some quantum channels are irreversible, such as dephasing channels and Markovian or non-Markovian noise processes, which cause information loss and making \(\rho_{\rm in}\) unrecoverable. To address this, we introduce the following theorem for the CQPT, based on Choi matrices.
\end{remark}
    
\begin{theorem}\label{theo:2}
    (Choi-based CQPT): A trainable Choi matrix \( \mathbf{J}_{\mathcal{E}} \) is a \( 4^N \times 4^N \) positive semi-definite matrix satisfying \( {\rm Tr}_{\mathcal{Y}}(\mathbf{J}_{\mathcal{E}}) = \mathbf{I} \). It represents a quantum process \( \mathcal{E} \), which is a CPTP map, if the following condition holds
    \begin{equation}\label{eq:rhofrhoi2}
    \rho_{\rm f} \propto {\rm Tr}_{\mathcal{X}}
    \Big[\Big([\mathcal{E}(\rho_{\rm in})]^\top \otimes \mathbf{I}\Big)\mathbf{J}_\mathcal{E}^{+}
    \Big] = \rho_{\rm in},
    \end{equation}  
    where 
    \( \mathbf{J}_\mathcal{E}^{+} \) is the pseudoinverse of \( \mathbf{J}_\mathcal{E} \), and \( {\rm Tr}_{\mathcal{X}(\mathcal{Y})} \) represents the partial trace over the Hilbert space \( H_{\mathcal{X}(\mathcal{Y})} \). The proportionality symbol accounts for a normalization factor in \( \rho_{\rm f} \).
\end{theorem}
See Figure~\ref{fig:1} (c, d) for a graphical illustration of the Choi-based quantum process and Theorem \ref{theo:2}.

\begin{proof}
    See details in Appendix B.  
\end{proof}

\begin{remark}\label{rmCost} (Cost function and training process).
Similar to the previous section, the cost function $C(\mathbf{J}_{\mathcal{E}})$ is defined through the infidelity
\begin{equation}\label{eq:cost_2}
C(\mathbf{J}_\mathcal{E}) = 1 - \int_{W \in \rm{Haar}} \left| \mathrm{Tr} \left[
\left(\mathcal{E}(\rho_{\mathrm{in}})^\top \otimes \rho_{\mathrm{in}} \right)
\mathbf{J}_\mathcal{E}^{+}
\right] \right|^2
\end{equation}
    where $\rho_{\rm in} = W|\bm 0\rangle\langle\bm0|W^\dagger$ for an arbitrary 
    random unitary $W$.
\end{remark}

\begin{proof}
    See details in Appendix B.
\end{proof}

The training process begins by initializing an \( N \)-qubit system in the Hilbert space \( H_{\mathcal{X}} \) in the state \( |00\cdots 0\rangle \). The operator \( W \) is applied, followed by the quantum channel \( \mathcal{E} \) and the transpose operator to obtain \( \mathcal{E}(\rho_{\rm in})^\top \). The same \( \rho_{\rm in} \) is prepared in another \( N \)-qubit system in the Hilbert space \( H_{\mathcal{Y}} \). The joint system \( H_{\mathcal{X}} \otimes H_{\mathcal{Y}} \) then undergoes the application of \( \mathbf{J}_\mathcal{E}^+ \). Finally, \( W^\dagger \) is applied to the \( H_{\mathcal{Y}} \) subspace, and the output state is measured in the \( |\bm 0\rangle \) basis, resulting in the probability \( p(00\cdots 0) \). To optimize \( \mathbf{J}_{\mathcal{E}} \), we minimize \( C(\mathbf{J}_\mathcal{E}) \) until it converges, ensuring that \( \mathbf{J}_{\mathcal{E}} \) accurately represents the Choi matrix of the quantum channel \( \mathcal{E} \). Further details are provided in Figure~\ref{fig:11} in Appendix B, where we present a conceptual, circuit-based schematic to clarify the logical flow of the CQPT algorithm. We emphasize that this schematic does not correspond to a physically implementable quantum circuit at present, due to several technical limitations. In particular, operations such as the exact matrix transpose cannot be directly realized as quantum gates, and the controlled implementation of general noise channels on quantum hardware remains challenging \cite{Brand2024,Liu2025simulationofopen,PRXQuantum.5.020332}. Therefore, the scheme here is intended solely as a conceptual illustration to aid understanding of the algorithmic structure, rather than as an experimental prescription.

\section{Numerical results}
\subsection{Benchmarking Haar random unitary quantum gates}
In this subsection, we apply the CQPT to benchmark Haar random unitary quantum gates. Given a process $\mathcal{E} \equiv U_{\rm Haar}$, corresponding to a Haar random unitary gate \cite{mezzadri2007generaterandommatricesclassical}, we construct trainable Kraus operators $\Bbbk = (K_1, K_2, \dots, K_\mathsf{k})^\top$.
We have
$\rho_\mathcal{E}= U_{\rm Haar} \rho_{\rm in} U_{\rm Haar}^{\dagger}$, and 
$\rho_\Bbbk = \sum_{l=1}^{\mathsf{k}} K_l \rho_{\rm in} K_l^{\dagger}$, where $\rho_{\rm in} = W|\bm 0\rangle\langle\bm 0|W^\dagger$.
Following Theorem \ref{theo:1}, the final state is given by
\begin{equation}\label{eq:rhof}
\rho_{\rm f} = U_{\rm Haar}^{\dagger} \left( \sum_{l=1}^{\mathsf{k}} K_l \rho_{\rm in} K_l^{\dagger} \right) U_{\rm Haar}.
\end{equation}

To optimize \( \Bbbk \), we generate \( 6^N \) random input states \( \rho_{\rm in} \) and use a single measurement, \( M = |\bm{0}\rangle\langle \bm{0}| \), for training \cite{PhysRevLett.130.150402,SurawyStepney2022projectedleast}. For testing, we generate another set of \( 6^N \) random input states. In general, both the training and testing sets can be chosen arbitrarily. The cost function in Equation.~(\ref{eq:cost}) is evaluated by approximating the integral with a summation over the training set. \( \Bbbk \) is then iteratively updated using Equation~(\ref{eq:optimizer}) until convergence.  

\begin{figure}[t]
    \centering
    \includegraphics[width=\columnwidth]{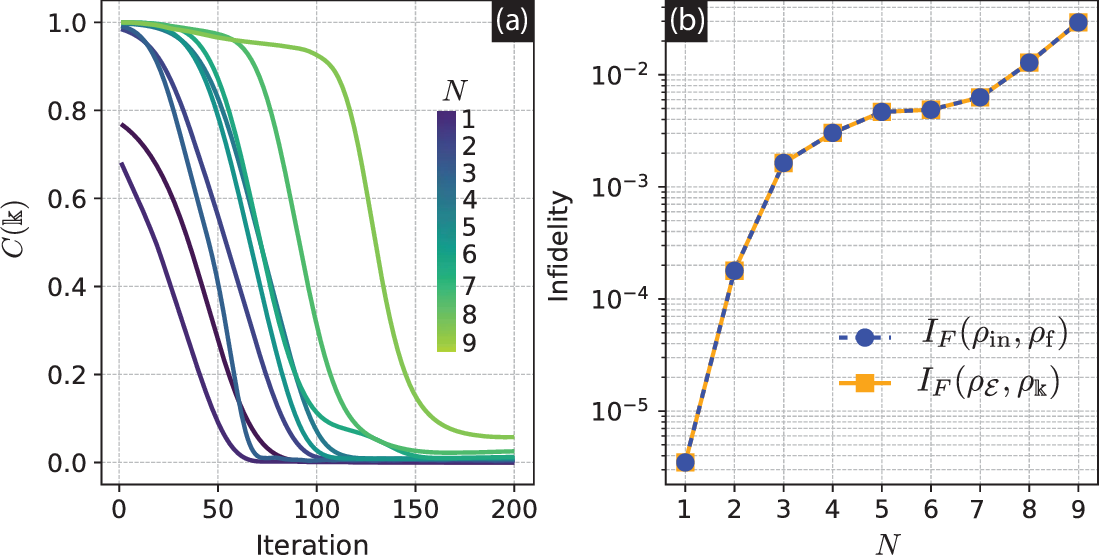}
    \caption{\textbf{Benchmarking Haar random unitary gates.}
    (a) Convergence of the cost function \(C(\Bbbk)\) over iterations.  
    (b) Average infidelity as a function of the qubit numbers \(N\).  
    }
    \label{fig:2}
\end{figure}

\begin{figure*}[t]
    \centering
    \includegraphics[width=0.8\textwidth]{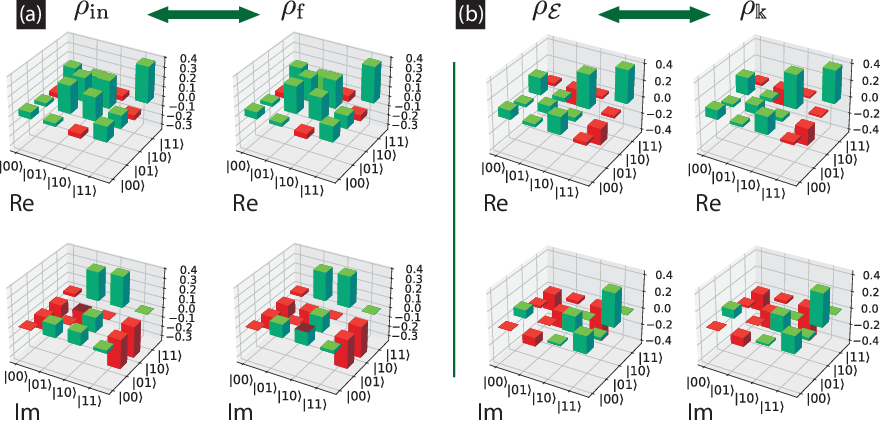}
    \caption{\textbf{Benchmarking Haar random unitary gates.}
    (a) Comparison of the initial and final states, \(\rho_{\rm in}\) and \(\rho_{\rm f}\), for $N = 2$, where $\rho_{\rm in}$ is selected from the testing set.
    Here, Re and Im denote the real and imaginary parts, respectively. Green indicates positive values, while red represents negative values.
    (b)
    Comparison of intermediate quantum states \(\rho_\mathcal{E}\) and \(\rho_\Bbbk\) to evaluate the transformation accuracy.  
    }
    \label{fig:3}
\end{figure*}

In our implementation, we use the number of Kraus terms \(\mathsf{k} = 2^N \ll 4^N\) (corresponding to a full-rank process). All numerical simulations are performed on a workstation equipped with an Intel Core i9-10920X processor (128 GB of RAM) and an NVIDIA GPU A6000.

Figure ~\ref{fig:2}(a) shows the cost function as a function of the iteration number for different \( N \). In general, the cost function gradually decreases, which indicates effective learning of the quantum channel. As \( N \) increases, the convergence rate slows down and more iterations are required to reach convergence.

Figure~\ref{fig:2}(b) compares the average infidelity across different \( N \). The infidelity between quantum states \(\rho\) and \(\sigma\) is defined as 
\begin{equation}
    I_F(\rho,\sigma) = 1 - {\rm Tr}\big(\sqrt{\sqrt{\rho}\sigma\sqrt{\rho}}\big),
\end{equation}  
where smaller values indicate higher similarity.
We evaluate the average infidelities \(I_F(\rho_{\rm in}, \rho_{\rm f})\) and \(I_F(\rho_{\mathcal{E}}, \rho_{\Bbbk})\) over \(6^N\) testing states. The results show that both infidelities closely match and increase with \(N\), suggesting that while the transformation remains accurate, deviations accumulate as \(N\) grows. Additionally, the numerical imprecision creates minor but unavoidable errors, which gradually increase the infidelity.

\begin{figure}[b]
    \centering
    \includegraphics[width= \columnwidth]{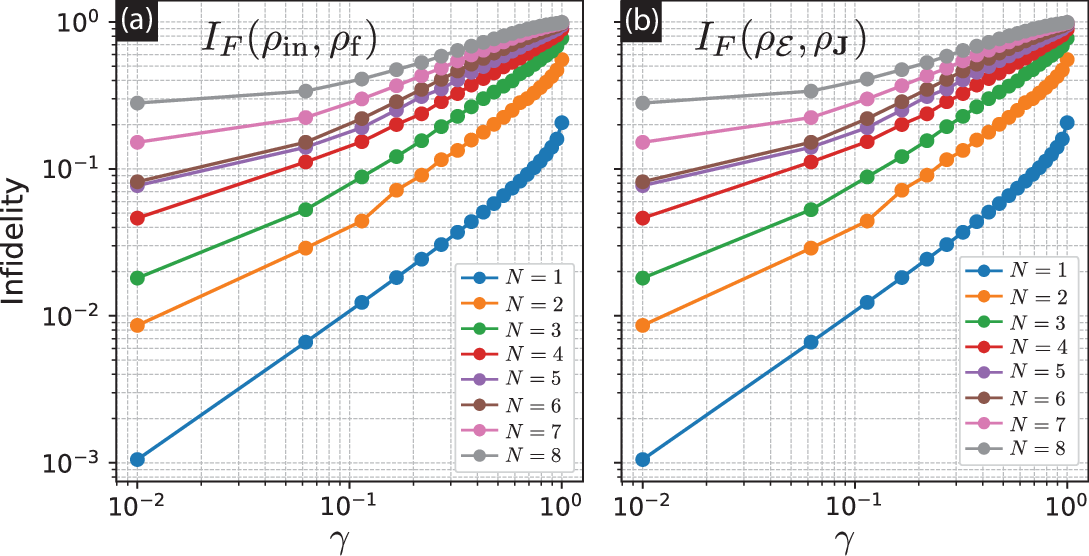}
    \caption{\textbf{CQPT for dephasing noise process.}
    (a) Average infidelity \(I_F(\rho_{\rm in}, \rho_{\rm f})\) versus \(\gamma\). (b) Average infidelity \(I_F(\rho_\mathcal{E}, \rho_{\mathbf{J}})\) versus \(\gamma\), showing the effect of noise on the optimization performance. 
    }
    \label{fig:4}
\end{figure}

Figure~\ref{fig:3} compares the density matrices for \(N = 2\) as an illustration. A random state \(\rho_{\rm in}\), chosen from the testing set, and its evolved state \(\rho_{\rm f}\) are shown in Figure~\ref{fig:3}(a), demonstrating a consistent quantum transformation. Likewise, Figure~\ref{fig:3}(b) compares \(\rho_\mathcal{E}\) and \(\rho_\Bbbk\), showing strong agreement, which confirms that the optimized quantum channel \(\Bbbk\) effectively reconstructs the target process \(\mathcal{E}\). The close match in both cases validates the accuracy of the reconstruction and the reliability of the method.

\subsection{Dephasing noise process}
Dephasing noise is a process that progressively erases coherence between quantum states due to environmental interactions. In multi-qubit systems, each qubit dephases independently, primarily affecting phase information while preserving state populations. The dephasing process for a qubit \(j\) is given by  
\begin{equation}  
    \mathcal{E}_j^{\rm deph}(\rho) = \frac{1 + \sqrt{1 - \gamma}}{2} \rho + \frac{1 - \sqrt{1 - \gamma}}{2} \sigma_z \rho \sigma_z,  
\end{equation}  
where \(\rho\) is the qubit state, \(\sigma_z\) is the Pauli-Z operator, and \(\gamma \in [0, 1]\) is the noise strength. At \(\gamma = 0\), there is no noise, and as \(\gamma\) increases, dephasing suppresses coherence and drives the state toward a maximally mixed state.

The total dephasing process on an $N$-qubit system is given by the product of the dephasing operations on each qubit
\begin{equation}
    \mathcal{E}_{\rm{total}}(\rho) = \bigotimes_{j=1}^{N} \mathcal{E}_j^{\rm deph}(\rho).
\end{equation}

We employ the Choi-based CQPT method. The process begins with the initial state \(\rho_{\rm in} = W|\bm{0}\rangle\langle\bm{0}| W^\dagger\), which evolves under the map \(\rho_{\mathcal{E}} = \mathcal{E}_{\rm total}(\rho_{\rm in})\). To reconstruct \(\mathcal{E}_{\rm total}\), we optimize a variational Choi matrix \(\mathbf{J}_{\mathcal{E}}\), which is initialized randomly. 
The pseudoinverse \(\mathbf{J}_{\mathcal{E}}^{+}\) is computed numerically at each optimization step using the Moore-Penrose pseudoinverse. For unitary processes, this pseudoinverse reduces to the ordinary matrix inverse. A conceptual circuit-based illustration of the CQPT framework is provided in Fig.~\ref{fig:11} and Appendix~B.
The training procedure involves generating \(6^N\) initial states \(\rho_{\rm in}\) and evaluating the cost function as defined in Equation~(\ref{eq:cost_2}). When the cost function \(C(\mathbf{J}_\mathcal{E}) \) reaches its minimum, we obtain the reconstructed state \(\rho_{\mathbf{J}} \approx \rho_\mathcal{E} \), where
\(
\rho_{\mathbf{J}} = {\rm Tr}_{\mathcal{X}} \Big[\big(\rho_{\rm in}^\top \otimes \mathbf{I} \big) \mathbf{J}_{\mathcal{E}}\Big].
\)
The numerical results are obtained using classical matrix operations as a proof of concept, rather than through an explicit circuit-level realization.

We evaluate the efficiency at different noise levels by calculating the average infidelity \(I_F(\rho_{\rm in}, \rho_{\rm f})\) and \(I_F(\rho_{\mathcal{E}}, \rho_{\mathbf{J}})\) over \(6^N\) random states from the testing set. Figure~\ref{fig:4} shows that for small noise parameters \(\gamma\), the optimized \(\mathbf{J}_\mathcal{E}\) accurately reconstructs the dephasing process. As \(\gamma\) increases, the recovery becomes less precise, leading to a larger deviation between the reconstructed and target states. Despite this, even for larger \(N\), the process remains effectively learned, although it becomes more challenging. These results indicate that higher noise levels and larger systems increase the complexity of the learning process, thereby reducing the accuracy of process tomography.

We further compare the evolved density matrices \(\rho_{\mathcal{E}}\) and \(\rho_{\mathbf{J}}\) for \(N = 2\) at low noise (\(\gamma = 0.01\)) and high noise (\(\gamma = 0.95\)). As shown in Figure~\ref{fig:5}(a) at low noise levels, the two states match closely for both real and imaginary parts. However, as the noise level increases, as in Figure~\ref{fig:5}(b), the difference between \(\rho_{\mathcal{E}}\) and \(\rho_{\mathbf{J}}\) becomes more pronounced, which can be seen clearly in the imaginary part, indicating reduced accuracy in the reconstruction process.

\begin{figure*}[t]
    \centering
    \includegraphics[width=0.8\textwidth]{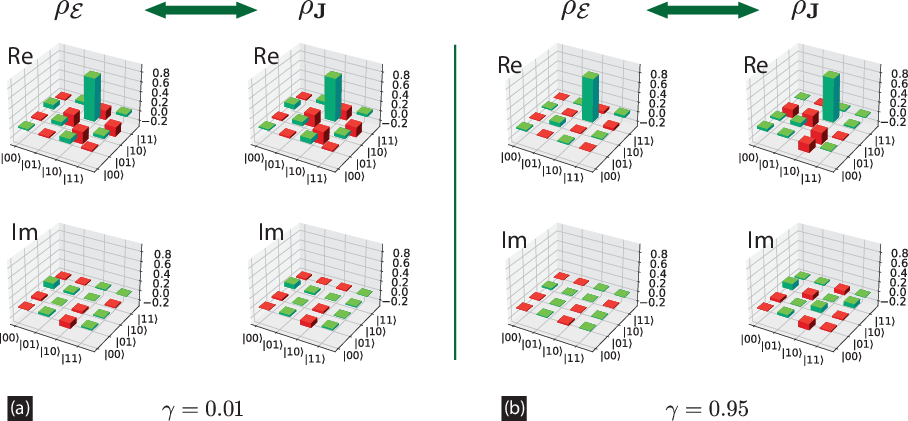}
    \caption{\textbf{CQPT for dephasing noise process.}
    (a) Comparison of quantum states \(\rho_\mathcal{E}\) and \(\rho_{\mathbf{J}}\) at \(\gamma = 0.01\), demonstrating a close match. (b) Comparison at \(\gamma = 0.95\), highlighting the growing deviation between the quantum states at higher noise.
    }
    \label{fig:5}
\end{figure*}

\subsection{Depolarizing and amplitude-damping processes}
\label{subsec:noise_models}

We further benchmark the performance of CQPT under the depolarizing noise and amplitude-damping noise. These channels capture distinct physical decoherence mechanisms and are widely used as standard test cases in quantum information.

\textbf{Depolarizing channel.} The $N$-qubit depolarizing channel is defined as
\begin{equation}
\mathcal{E}^{\mathrm{depol}}(\rho)
=
(1-p)\rho
+
p\,\frac{\mathbf{I}}{2^N},
\label{eq:depolarizing}
\end{equation}
where $p\in[0,1]$ denotes the depolarizing strength, $\rho$ is the input quantum state.

\textbf{Amplitude-damping noise.}
Amplitude-damping noise models irreversible energy relaxation processes. The 
amplitude-damping channel applied to the $j$th qubit with rate $\gamma\in[0,1]$ is defined by
\begin{equation}
\mathcal{E}_j^{\mathrm{damp}}(\rho)
=
\sum_{k=0}^{1} K_k \rho K_k^{\dagger},
\end{equation}
with Kraus operators
\(
K_0 =
\begin{pmatrix}
1 & 0 \\
0 & \sqrt{1-\gamma}
\end{pmatrix},
\;
K_1 =
\begin{pmatrix}
0 & \sqrt{\gamma} \\
0 & 0
\end{pmatrix},
\)
which satisfy $\sum_k K_k^{\dagger}K_k=\mathbf{I}$.  
For an $N$-qubit system, the total amplitude damping channel gives
\begin{equation}
\mathcal{E}^{\mathrm{damp}}_{\rm total}(\rho)
=
\bigotimes_{j=1}^{N} \mathcal{E}^{\mathrm{damp}}_{j}(\rho) .
\end{equation}

Figure~\ref{fig:5b} shows the CQPT performance for depolarizing and amplitude-damping processes, with the average infidelity $I_F(\rho_{\mathrm{in}}, \rho_f)$ plotted versus the depolarizing strength $p$ in panel (a) and the amplitude-damping rate $\gamma$ in panel (b). In both cases, $I_F(\rho_{\mathrm{in}}, \rho_f)$ increases monotonically with the noise strength, similar to the depolarizing case discussed above, reflecting the increased difficulty of channel reconstruction at higher noise levels. For fixed noise parameters, the infidelity increases with the number of qubits $N$, indicating the growing complexity of multi-qubit non-unitary dynamics.

\begin{figure}[b]
    \centering
    \includegraphics[width= \columnwidth]{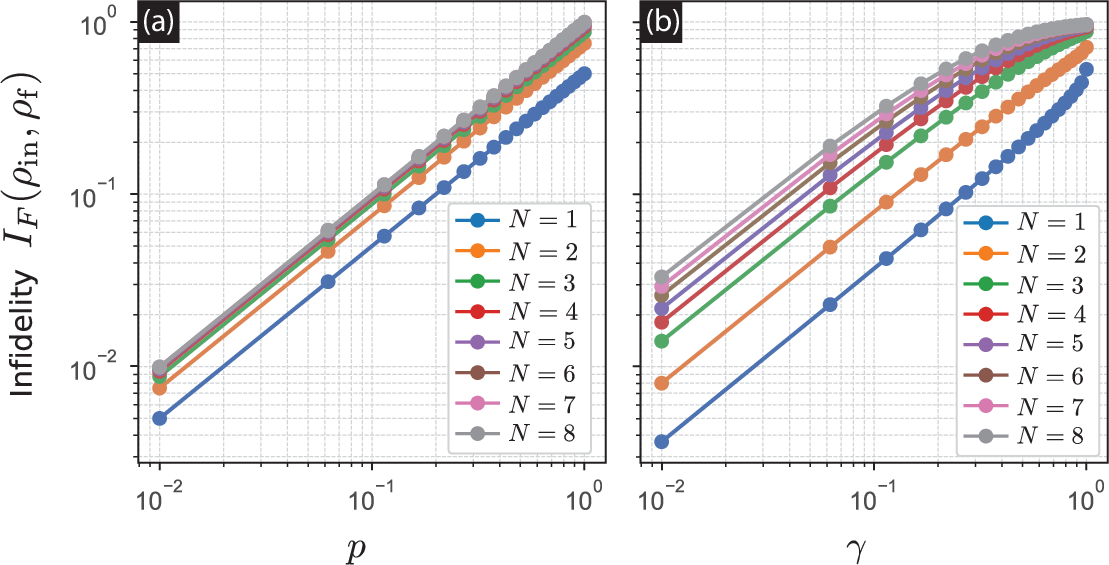}
    \caption{\textbf{CQPT for depolarizing and amplitude-damping process.}
    (a) Average infidelity \(I_F(\rho_{\rm in}, \rho_{\rm f})\) versus the depolarizing strength \(p\). Average infidelity \(I_F(\rho_{\rm in}, \rho_{\rm f})\) versus the amplitude-damping rate \(\gamma\). 
    }
    \label{fig:5b}
\end{figure}

\subsection{Time-homogeneous and time-inhomogeneous noises}
In quantum systems, noise can be categorized by how it changes over time. The following cases describe two types of dephasing noise, defined through
\begin{equation}
\gamma = \left\{
\begin{array}{ll}
1 - e^{-2\beta t}, & \textrm{time-homogeneous} \\[1ex]
1 - e^{-\beta t^2}, & \textrm{time-inhomogeneous}
\end{array}
\right.
\end{equation}
where \( t \) denotes time, and \( \beta \) is the dephasing rate. The noise strength \( \gamma \) increases from 0 to 1 as \( t \) changes from 0 to \(\infty \).  

The first case represents time-homogeneous noise, where the noise strength increases exponentially over time. The noise rate remains constant, leading to a Markovian process where the noise evolution depends only on the current state. Such behavior is typical of many natural decay processes, such as thermal dephasing or spontaneous emission, where the loss of coherence occurs steadily over time.  

In contrast, the second case describes time-inhomogeneous noise. Here, the noise follows a quadratic time dependence, growing slowly at first and then accelerating as time evolves. This results in a non-Markovian process, where the rate of dephasing changes over time. This behavior is typical of systems with diffusive noise or time-dependent fluctuations.

Numerically, we examine a 2-qubit system and set \(\beta = 0.1\). We first reconstruct the time-dependent states \(\rho_\mathcal{E}(t)\) and \(\rho_{\mathbf{J}}(t)\) for both noise cases using the CQPT. We then calculate the expectation value \(\langle \sigma_x(t)\rangle\) of the first qubit as
\begin{eqnarray}
\langle \sigma_x(t) \rangle_\mathcal{E} &=& \mathrm{Tr} \left[ \big(\sigma_x\otimes I\big) \rho_\mathcal{E}(t) \right], \\ 
\langle \sigma_x(t) \rangle_\mathbf{J} &=& \mathrm{Tr} \left[ \big(\sigma_x\otimes I\big) \rho_\mathbf{J}(t) \right].
\end{eqnarray}
The purpose of this example is to demonstrate the applicability of the CQPT framework to multi-qubit non-unitary channels under different dynamical regimes.

\begin{figure}[t]
    \centering
    \includegraphics[width=\columnwidth]{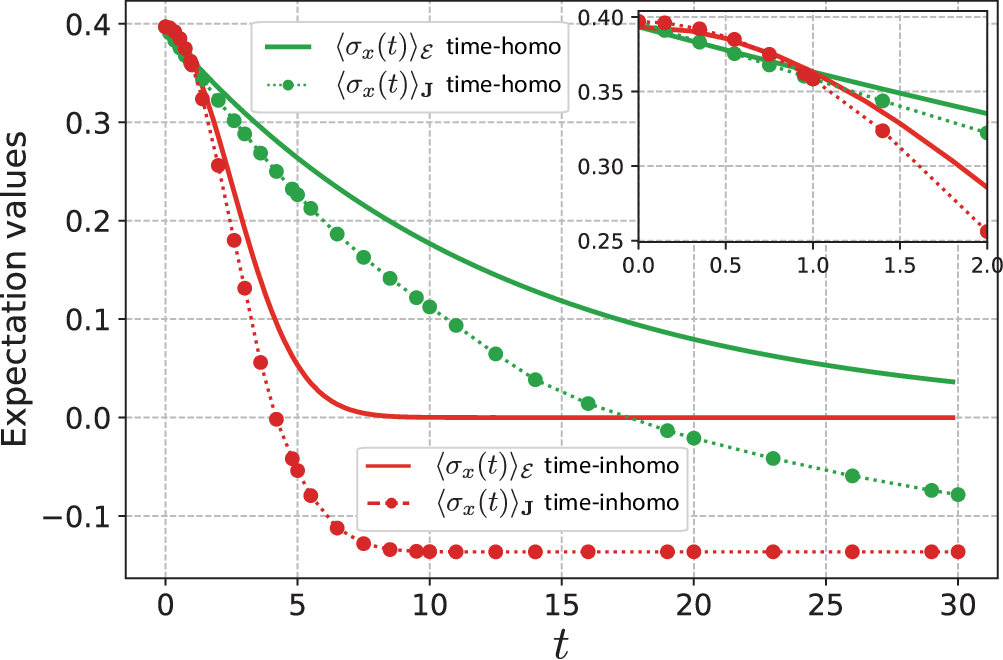}
    \caption{\textbf{Expectation values}. Plot of the expectation values
$\langle \sigma_x(t) \rangle_\mathcal{E}$ and
$\langle \sigma_x(t) \rangle_\mathbf{J}$ for time-homogeneous and time-inhomogeneous noise processes. The plot compares the expectation values calculated from the evolved state $\rho_{\mathcal{E}}$ and from the reconstructed state $\rho_{\mathbf{J}}$. Inset: plot of the small-time regime, highlighting the close agreement between the true and reconstructed expectation values.
}
    \label{fig:6}
\end{figure}

Figure~\ref{fig:6} displays the expectation values over time for both time-homogeneous and time-inhomogeneous noise, comparing the true evolved state with the reconstructed process. For time-homogeneous noise, the expectation values decay exponentially, indicating a constant-rate decoherence typical of Markovian processes like thermal dephasing. In contrast, for time-inhomogeneous noise, the decay starts slowly and then accelerates quadratically, reflecting the presence of a non-Markovian process, such as diffusive noise or time-dependent fluctuations.

In the short-time regime, shown in the inset of Fig.~\ref{fig:6}, the reconstructed expectation values closely follow the theoretical predictions, demonstrating that the CQPT framework accurately captures the underlying dynamics when noise effects remain moderate. At longer times, deviations between the reconstructed and true expectation values become more pronounced, particularly for time-inhomogeneous noise, where the accumulated noise strength significantly degrades the recoverability of the quantum state. This behavior is expected, as strong and prolonged decoherence reduces the available information for reliable process reconstruction.

Despite this limitation, the results indicate that the proposed CQPT method can reliably reconstruct quantum dynamics at a short-time limit, providing a useful proof of principle for quantum process characterization under realistic conditions.

\section{Comparison}
We compare our approach with previous studies \cite{volya2024fastquantumprocesstomography,PhysRevLett.130.150402} to assess its efficiency. These previous works utilize a Riemannian gradient-descent-based quantum process tomography method, which we also adopt for optimization. However, our training process differs significantly.  

Previous studies, which we refer to as measurement-based QPT (MQPT), directly compare the measurement outcomes between \( \rho_\mathcal{E} \) and \( \rho_\Bbbk \) 
obtained from different processes, using a cost function defined as  
\begin{equation}\label{eq:mqpt}  
    C(\Bbbk) = \sum_{ij} \Big( {\rm Tr} \big[ M_j (\rho^{i}_{\mathcal{E}} - \rho^{i}_{\Bbbk}) \big] \Big)^2,  
\end{equation}  
where \( i \in [1, \mathsf{n}] \) labels \( \mathsf{n} \) different input states \( \rho_{\rm in}^1, \dots, \rho_{\rm in}^\mathsf{n} \), and \( j \in [1, \mathsf{p}] \) indexes the \( \mathsf{p} \) elements of the POVM.
In contrast, our approach, referred to as compilation-based QPT (CQPT), uses a quantum compilation process for QPT, rather than performing direct state comparison. This key difference allows a more structured and potentially more efficient reconstruction of quantum processes.

We analyze an \( N \)-qubit system using a Haar random unitary process for comparison. The dataset consists of \( \mathsf{n} = 6^N \) initial states \( \rho_{\rm in} \). The MQPT method requires applying a measurement operator \( M \) with \( \mathsf{p} = 6^N \) measurements to each transformed state \( \rho_{\mathcal{E}} \) and \( \rho_{\Bbbk} \) \cite{volya2024fastquantumprocesstomography,PhysRevLett.130.150402}. This results in a total of \( 6^N \times 6^N \) measurements per learning iteration.
In contrast, the CQPT  applies \( \mathcal{E} \) and \( \Bbbk \) to \( \rho_{\rm in} \) and requires only a single measurement per state, reducing the total to \( 6^N \) measurements per iteration.

Correspondingly, the total measurement time for the MQPT is given by
\begin{equation}
    T_{\rm{MQPT}} = c_1 \cdot (6^{2N}),
\end{equation}
which grows exponentially with \( N \), where \( c_1 \) is the time per measurement. 
For the CQPT, the measurement time scales as
\begin{equation}
    T_{\rm{CQPT}} = c_2 \cdot 6^N,
\end{equation}
where \( c_2 \) is the measurement time per state. This follows a smaller exponential scaling, thereby reducing the overall computational cost.

\begin{figure*}[t]
    \centering
    \includegraphics[width=0.7\linewidth]{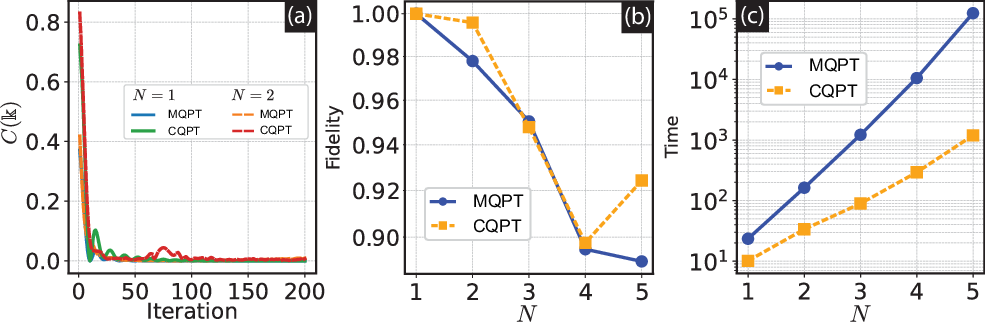}
    \caption{\textbf{Comparing the MQPT and CQPT methods}.
   (a) Cost function \( C(\Bbbk) \) versus the number of iterations for \( N = 1 \) and \( N = 2 \) qubits, comparing the two methods. For the MQPT, \( C(\Bbbk) \) is given by Eq.~(\ref{eq:mqpt}), and for the CQPT, \( C(\Bbbk) \) is given by Eq.~(\ref{eq:ck_our}).
   (b) Plot of the fidelity as a function of the number of qubits \( N \) for both methods.
    (c) Execution time as a function of \( N \).
    }
    \label{fig:7}
\end{figure*}

Figure~\ref{fig:7} presents a direct comparison between the MQPT and CQPT methods and shows that the efficiency advantage of the CQPT becomes more pronounced as the system size \(N\) increases. In the MQPT, the number of required measurements grows rapidly, leading to a substantial increase in time complexity. In contrast, the CQPT avoids redundant measurements by design, which significantly improves its scalability for larger quantum systems. In particular, although both the MQPT and CQPT employ Riemannian gradient descent for optimization, the convergence behavior differs markedly. As shown in Fig.~\ref{fig:7}(a), the number of iterations required for convergence grows more slowly in the CQPT, resulting in faster convergence. Moreover, for system sizes up to \(N = 5\), the CQPT achieves fidelity that is comparable to, or even higher than, that of MQPT, as illustrated in Fig.~\ref{fig:7}(b).

Finally, we compare the overall convergence time. In the CQPT, it is dominated by the optimization complexity of the process representation, whereas in the MQPT it is further amplified by the exponential growth of measurement counts.
As a result, the CQPT significantly reduces the execution time compared to the MQPT as shown in Figure~\ref{fig:7}(c). Moreover, as the number of qubits increases, the execution time in the CQPT grows at a slower rate, highlighting its efficiency.

For the memory complexity, the MQPT requires storing both the process representation $\mathcal{O}(2^N\cdot 2^N)$ and extensive measurement outcomes, leading to $\mathcal{O}(4^N \cdot 6^N)$. The CQPT, by contrast, stores only the trainable operators $\mathcal{O}(4^N)$ and minimal final state data, achieving $\mathcal{O}(4^N)$. This efficiency, free from the $6^N$ dependent term, makes the CQPT significantly more memory-efficient and practical for large-scale quantum process tomography.

In the MQPT, the cost function in Eq.~(\ref{eq:mqpt}) depends on discrepancies between the full POVM measurement statistics of \( \rho_\mathcal{E} \) and \( \rho_\Bbbk \). This leads to a highly noisy optimization landscape, and as a result, the optimization typically exhibits slower convergence. In contrast, the CQPT formulates the learning problem as a compilation task using single-shot measurements. This formulation effectively compresses the optimization landscape into a lower-dimensional manifold governed by the trainable Kraus operators or the Choi matrix.

\section{Conclusion}
We introduced a compilation-based quantum process tomography (CQPT) framework optimized using the Riemannian gradient descent. This approach models an unknown quantum process through a set of trainable operators, such as the Kraus operators or the Choi matrix, which enables efficient extraction of process information. By enforcing that the output state returns to its initial reference state, it avoids the need for full tomography or full-rank measurement bases, significantly reducing both measurement complexity and computational cost.  

The CQPT improves the runtime efficiency compared to previous methods using the same Riemannian gradient descent \cite{volya2024fastquantumprocesstomography,PhysRevLett.130.150402}. A key feature is that it uses single-shot measurements, which removes the requirement of full-base measurements, i.e, tensor products of eigenvalues of Pauli matrices, substantially reduces the overall measurement cost. As a result, the CQPT is well suited for larger quantum systems and near-term quantum hardware with limited coherence times. 
Furthermore, the CQPT improves precision by reducing statistical and readout errors. The framework applies to both unitary and noisy quantum processes, making it suitable for quantum gates validation, circuits benchmarking, and noise characterization.

However, the CQPT has some limitations. While it can handle noisy quantum processes, strong noise can still reduce reconstruction accuracy. So far, state preparation and measurement (SPAM) errors remain a challenge in the QPT. In the CQPT, single-shot measurements help reduce the measurement-related errors, however, state preparation errors can still affect the training process. In future works, this issue may be addressed using measurement-error mitigation techniques \cite{PhysRevResearch.6.013187}, or learning-based calibration strategies for state preparation \cite{AlcaldePuente2025learningfeedback,Wang2025}.

Despite these challenges, the CQPT provides an efficient and accurate approach to the QPT, making it a valuable tool for advancing quantum technologies.

\section*{Appendix A: Implementation of Kraus-based CQPT (Theorem \ref{theo:1})}
The Kraus-based CQPT process begins by preparing a set of training quantum states, \(\rho_{\rm in}^{i} = W_i \rho_0 W_i^\dagger\) for \(i \in [1, \mathsf{n}]\), where \(\rho_0\) is used as the reference state. The process then applies the Kraus operators to evolve the initial state, followed by an inverse operation, \(\mathcal{E}^{-1}\). Finally, applying \(W_i^\dagger\) restores the reference state \(\rho_0\). This approach is particularly effective when \(\mathcal{E}\) is invertible, such as in the unitary processes. 

For example, consider a random unitary process \(\mathcal{E} = U\). We can also choose \(\Bbbk = K\), which is a unitary operator. Theorem \ref{theo:1} gives 
\begin{eqnarray}
\rho_0 &\to& W_i \rho_0 W_i^\dagger 
\to K \big(W_i \rho_0 W_i^\dagger\big) K^\dagger \nonumber\\
&\to& U^\dagger \Big[ K\big(W_i \rho_0 W_i^\dagger\big) K^\dagger \Big] U \nonumber\\
&\to& W_i^\dagger \Big[ U^\dagger \big[ K(W_i \rho_0 W_i^\dagger) K^\dagger\big] U \Big] W_i.
\end{eqnarray}

The training is performed until \(K = U\), at which point the final state becomes \(\rho_0\). Without loss of generality, we can fix \(\rho_0\) to be either a ground state or the computational basis state \(|\bm{0}\rangle\). If \(\rho_0 = |\bm{0}\rangle\langle \bm 0|\), then measuring the final state in the basis \(M_{\bm{0}}\) must yield \(p_{\bm{0}} = 1\) for all \(W_i\).  
To achieve optimal training, i.e., \(K = U\), we define a cost function based on the average value of \(p_{\bm{0}}\) over a random set of \(W_i\), as given in Equation~(\ref{eq:cost}).  
This CQPT scheme can be implemented in a quantum circuit, as illustrated in Figure~\ref{fig:8}.  

\begin{figure}[t]
    \centering
    \includegraphics[width=\columnwidth]{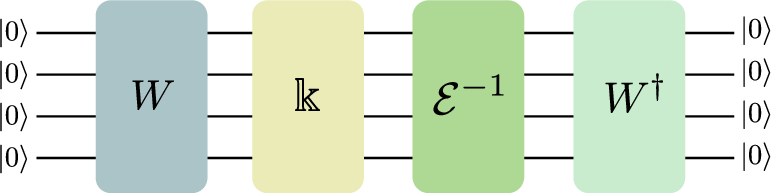}
    \caption{Implementation of the Kraus-based CQPT.}
    \label{fig:8}
\end{figure}

\section*{Appendix B: Implementation of Choi-based CQPT (Theorem \ref{theo:2})}\label{app:b}
\subsection*{Proof of Theorem \ref{theo:2}}
\label{app:b1}
To prove Theorem \ref{theo:2}, we assume that \(\mathbf{J}_{\mathcal{E}}\) is the corrected Choi matrix representing \(\mathcal{E}\) and show that the left-hand side (LHS) and right-hand side (RHS) of Equation~(\ref{eq:rhofrhoi2}) are equivalent.

First, we derive $\mathcal{E}(\rho_{\rm in})$, where
\begin{eqnarray}\label{eq:Erhoi}
\mathcal{E}(\rho_{\rm in}) &=& \mathrm{Tr}_{\mathcal{X}} \Big[ (\rho_{\rm in}^{\top} \otimes \mathbf{I}) \mathbf{J}_{\mathcal{E}} \Big] \nonumber \\
&=& \sum_k \Big(\langle k|\otimes \mathbf{I}\Big) 
\Big[ (\rho_{\rm in}^{\top} \otimes \mathbf{I}) \mathbf{J}_{\mathcal{E}} \Big]
\Big(|k\rangle\otimes \mathbf{I}\Big) \nonumber\\
&=& \sum_k \mathbf{T}_k^\dagger 
\Big[ (\rho_{\rm in}^{\top} \otimes \mathbf{I}) \mathbf{J}_{\mathcal{E}} \Big]
\mathbf{T}_k,
\end{eqnarray}
where $\mathbf{T}_k = |k\rangle\otimes \textbf{I} \in H_{\mathcal{X}}\otimes H_{\mathcal{Y}}$ for some orthonormal bases $\{|k\rangle\} \in H_{\mathcal{X}}$.

Next, we apply vectorization on both sides of Equation~(\ref{eq:Erhoi}), which gives
\begin{eqnarray}\label{eq:vecErhoi}
{\rm vec}[\mathcal{E}(\rho_{\rm in})] 
&=&
{\rm vec}\Big[\sum_k \mathbf{T}_k^\dagger 
\big[ (\rho_{\rm in}^{\top} \otimes \mathbf{I}) \mathbf{J}_{\mathcal{E}} \big]
\mathbf{T}_k\Big]
\quad \label{eq:a} \\
&=&
\sum_k {\rm vec}
\Big[\mathbf{T}_k^\dagger 
\big[ (\rho_{\rm in}^{\top} \otimes \mathbf{I}) \mathbf{J}_{\mathcal{E}} \big]
\mathbf{T}_k
\Big]
\quad \label{eq:b} \\
&=&
\sum_k
\Big(\mathbf{T}_k^\top \otimes \mathbf{T}_k^\dagger\Big)
{\rm vec}
\Big[(\rho_{\rm in}^{\top} \otimes \mathbf{I}) \mathbf{J}_{\mathcal{E}} \Big]
\quad \label{eq:c} \\
&=&
\sum_k
\Big(\mathbf{T}_k^\top \otimes \mathbf{T}_k^\dagger\Big)
\Big(\mathbf{J}_\mathcal{E}^\top \otimes \mathbf{I} \Big)
{\rm vec}
\Big[(\rho_{\rm in}^{\top} \otimes \mathbf{I}) \Big]
\quad \label{eq:d} \\
&=&
\mathbf{J}_{\mathcal{E}}^\top \ {\rm vec}[\rho_{\rm in}^\top]
\quad \label{eq:e}.
\end{eqnarray}
where from (\ref{eq:b}) to (\ref{eq:c}) and from (\ref{eq:c}) to (\ref{eq:d}) we used 
${\rm vec} [ABC] = (C^\top\otimes A){\rm vec}[B]$ \cite{PhysRevA.97.042322}.
To justify (\ref{eq:d}) $\to$ (\ref{eq:e}), we substitute \(\mathbf T_k=|k\rangle\otimes I\) into (\ref{eq:d}), which selects the \((k,k)\) block on \(H_{\mathcal X}\), the sum over \(k\) (in \ref{eq:d} ) therefore implements \(\mathrm{Tr}_{\mathcal X}\), and the resulting block-weighted sum
becomes \(\mathbf J_{\mathcal E}^{\top}\mathrm{vec}(\rho_{\rm in}^{\top})\) after vectorization.
Inversely, we also have 
\begin{equation}\label{eq:vecErhot}
    {\rm vec} [\mathcal{E}(\rho_{\rm in})^\top] = 
    \mathbf{J}_{\mathcal{E}}\ {\rm vec}[\rho_{\rm in}].
\end{equation}

Similarly, let the LHS of Equation~(\ref{eq:rhofrhoi2}) be
$\mathbf{L} = {\rm Tr}_{\mathcal{X}}\Big[\Big([\mathcal{E}(\rho_{\rm in})]^\top \otimes \mathbf{I}\Big)\mathbf{J}_\mathcal{E}^{+}
    \Big]$, which gives
\begin{equation}\label{eq:vecL}
    {\rm vec}[\mathbf{L}] = 
    (\mathbf{J}_{\mathcal{E}}^+)^\top
    {\rm vec}[\mathcal{E}(\rho_{\rm in})^\top].
\end{equation}
Substituting Equation~(\ref{eq:vecErhot}) into Equation~(\ref{eq:vecL}), we get 
\begin{eqnarray}\label{eq:vecL1}
    {\rm vec}[\mathbf{L}] &=& (\mathbf{J}_{\mathcal{E}}^+)^\top \mathbf{J}_{\mathcal{E}}\, {\rm vec}[\rho_{\rm in}] \nonumber \\
    &=& \mathbf{P} {\rm vec}[\rho_{\rm in}] \nonumber \\
    &=& {\rm vec}[\rho_{\rm in}] \equiv {\rm vec}[\mathbf{R}].
\end{eqnarray}
where we used 
$(\mathbf{J}_{\mathcal{E}}^+)^\top \mathbf{J}_{\mathcal{E}} = (\mathbf{J}_{\mathcal{E}}^+) \mathbf{J}_{\mathcal{E}} = \mathbf{P}$ is a projection matrix, i.e., $\mathbf{PX} = \mathbf{X}, \forall \mathbf{X}$, and $\mathbf{R} = \rho_{\rm in}$, which stands for the RHS. As a result, Equation~(\ref{eq:rhofrhoi2}) is proven. 

\begin{figure}[t]
    \centering
    \includegraphics[width=\columnwidth]{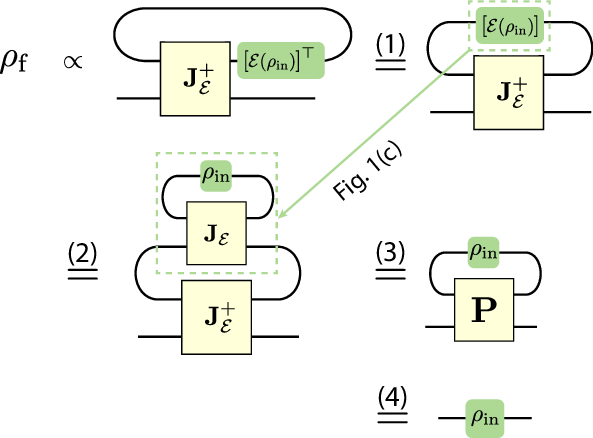}
    \caption{Graphical proof of Theorem \ref{theo:2}.}
    \label{fig:9}
\end{figure}

\subsection*{Graphical proof of Theorem \ref{theo:2}}
In Figure~\ref{fig:9}, we provide a graphical proof of Theorem \ref{theo:2}. Starting from Figure~\ref{fig:1}(d), we first transform $[\mathcal{E}(\rho_{\rm in})]^\top$ into $[\mathcal{E}(\rho_{\rm in})]$, and then apply the operation in Figure~\ref{fig:1}(c). Next, we replace $\mathbf{J}_\mathcal{E}^+\mathbf{J}_\mathcal{E}$ with $\mathbf{P}$. Finally, since $\rho_{\rm in}$ is a projection operator, we recover $\rho_{\rm in}$.

\subsection*{Proof of Remark \ref{rmCost}}
We prove that ${\rm Tr}\Big[
\Big(\mathcal{E}(\rho_{\rm in})^\top\otimes 
\rho_{\rm in}\Big)
\mathbf{J}_\mathcal{E}^{+}
\Big] = 1, \forall W$, when $\mathbf{J}_{\mathcal{E}}$ represents the Choi matrix for the quantum channel $\mathcal{E}$.
We start from 
\begin{equation}
    {\rm Tr}\Big[(A\otimes B)C\Big] = {\rm Tr}\Big[{\rm Tr}_{\mathcal{X}}\big[(A^\top\otimes \mathbf{I})C\big]B\Big],
\end{equation}
where $A\in H_{\mathcal{X}}, B\in H_{\mathcal{Y}}$
and $C\in H_{\mathcal{X}}\otimes H_{\mathcal{Y}}$.
Applying this identity to our case, we have
\begin{eqnarray}
    {\rm Tr}\Big[
    \big(\mathcal{E}(\rho_{\rm in})^\top \otimes \rho_{\rm in}\big) \mathbf{J}_\mathcal{E}^+
    \Big] 
    &=& {\rm Tr}_{\mathcal{Y}} \Big[
    {\rm Tr}_{\mathcal{X}} \Big[
    \big(\mathcal{E}(\rho_{\rm in})^\top \otimes \mathbf{I}\big) \mathbf{J}_\mathcal{E}^+
    \Big] \rho_{\rm in}
    \Big] \label{eq:ma} \\
    &=& {\rm Tr}_{\mathcal{Y}} \big[ \rho_{\rm in} \rho_{\rm in} \big] \label{eq:mb} \\
    &=& 1. \label{eq:mc}
\end{eqnarray}
where from Eq.~(\ref{eq:ma}) to Eq.~(\ref{eq:mb}) we used Eq.~(\ref{eq:rhofrhoi2}), and 
from Eq.~(\ref{eq:mb}) to Eq.~(\ref{eq:mc}) we used the fact that $\rho_{\rm in}$ is a pure state, i.e., $W|\bm0\rangle$.

\subsection*{Graphical proof of Remark \ref{rmCost}}
We illustrate a graphical proof of Remark \ref{rmCost} in Figure~\ref{fig:10}.

\begin{figure}[b]
    \centering
    \includegraphics[width=\columnwidth]{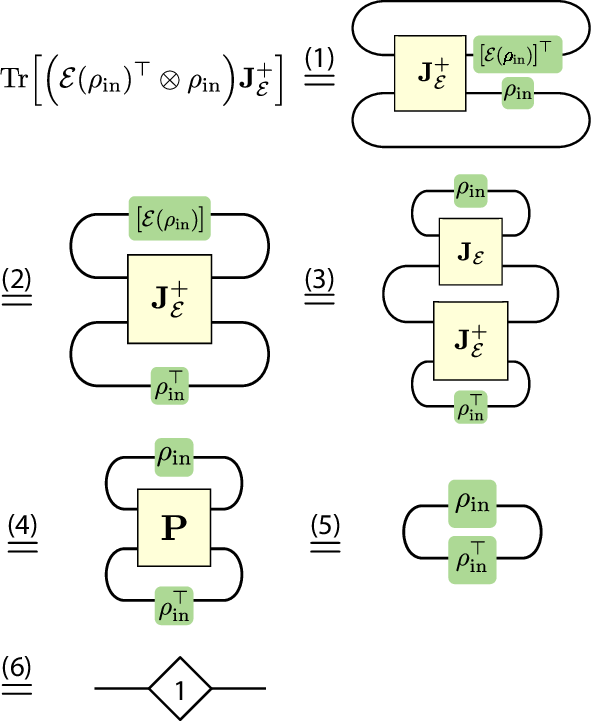}
    \caption{Graphical proof of Remark \ref{rmCost}.}
    \label{fig:10}
\end{figure}

\begin{figure}[t]
    \centering
    \includegraphics[width=\columnwidth]{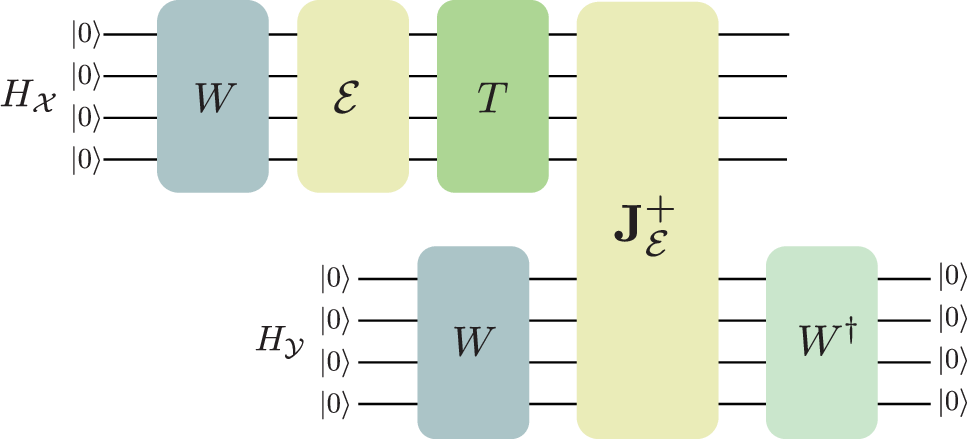}
    \caption{Conceptual circuit-based illustration of the Choi-based CQPT.}
    \label{fig:11}
\end{figure}

\begin{figure}[b]
    \centering
    \includegraphics[width=\columnwidth]{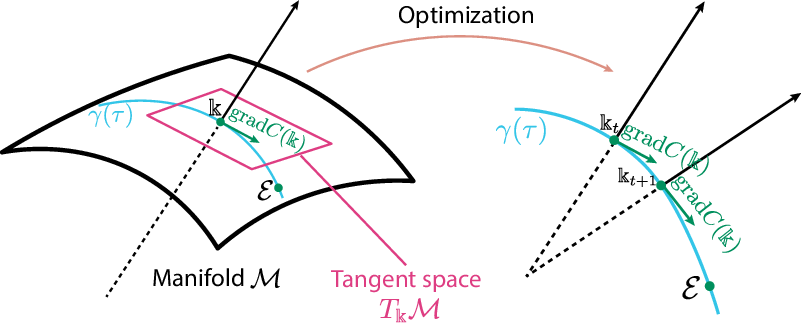}
    \caption{Illustration of the Riemannian optimization.}
    \label{fig:12}
\end{figure}

\subsection*{Conceptual circuit-based illustration}
To illustrate the cost function in Remark~\ref{rmCost}, we present a schematic representation of the operator flow within the CQPT framework, as shown in Fig.~\ref{fig:11}. In contrast to the unitary-channel case discussed earlier, the diagram here is purely conceptual and is intended to clarify the algebraic structure of the protocol. It does not represent a physically implementable quantum circuit, and a full hardware realization of the Choi matrix operations remains an open direction for future work.

First, we prepare a pure state \(\rho_{\rm in}=W|\mathbf 0\rangle\langle \mathbf 0|W^\dagger\) in \(H_{\mathcal X}\) and an identical copy in the \(H_{\mathcal Y}\) subsystem.
The state in \(H_{\mathcal X}\) is then subjected to the quantum channel \(\mathcal E\) followed by the transpose operator $T$.
As a result, the joint state becomes \(\mathcal E(\rho_{\rm in})^{\top}\otimes \rho_{\rm in}\).
Next, we apply the pseudoinverse Choi matrix $\mathbf{J}_{\mathcal{E}}^+$ to the joint system, followed by applying $W^\dagger$ to the $H_{\mathcal{Y}}$ subsystem, and then measure the probability $p_{\bm 0}$ of obtaining $|\mathbf{0}\rangle$. By optimizing $p_{\bm 0}$, we can refine the cost function to obtain the desired Choi matrix $\mathbf{J}_{\mathcal{E}}$.

We emphasize that the transpose operator $T$ is a purely formal algebraic operation that arises from the vectorization, and does not correspond to a physically implementable quantum gate. Therefore, a direct realization of this operation as a quantum circuit is nontrivial and lies beyond the capabilities of current quantum hardware, and remains an open direction for future research. The pseudoinverse Choi matrix \(\mathbf{J}_{\mathcal{E}}^{+}\) is instead modeled by a variational quantum circuit with trainable parameters. After optimization, a subsequent application of the pseudoinverse operation will recover the Choi matrix \(\mathbf{J}_{\mathcal{E}}\). 
In this work, the validity and feasibility of this approach are demonstrated through numerical simulations, where the Choi-based CQPT protocol is implemented using matrix operations as a proof of concept, rather than through an explicit circuit-level realization of \(\mathbf{J}_{\mathcal{E}}^{+}\).

\begin{table*}[t]
    \centering
    \scriptsize 
    \renewcommand{\arraystretch}{1.1} 
    \setlength{\tabcolsep}{4pt} 
    \begin{tabular}{@{}p{2.0cm} p{1.6cm} p{3.4cm} p{2.8cm} p{3cm} p{3cm}@{}}
        \toprule
        \textbf{Retraction Method} & \textbf{Order} & \textbf{Description} & \textbf{Complexity} & \textbf{Pros} & \textbf{Cons} \\
        \midrule
        QR-based & First & Uses QR decomposition to project onto the manifold. & \(\mathcal{O}(n^3)\) (QR) & Simple, efficient. & Lower accuracy. \\
        \midrule
        Polar-based & Second & Uses polar decomposition for retraction. & \(\mathcal{O}(n^3)\) (SVD) & Higher accuracy. & More expensive than QR. \\
        \midrule
        Cayley Transform & Second & Uses Cayley transform to approximate the exponential map. & \(\mathcal{O}(n^3)\) (Matrix inversion) & Preserves structure, works well for Lie groups. & Requires stability adjustments. \\
        \midrule
        Exponential Map & Exact & Computes the matrix exponential along the geodesic. & \(\mathcal{O}(n^3)\) (Matrix exp) & Most accurate. & Computationally expensive. \\
        \bottomrule
    \end{tabular}
    \caption{Comparison of retraction methods in Riemannian optimization.}
    \label{tab:retractions}
\end{table*}

\section*{Appendix C: Riemanian optimization}
Consider the Stiefel manifold, which consists of all $ d \times u $ complex matrices with orthonormal columns
$
\mathcal{M}(d,u) = \{ \Bbbk \in \mathbb{C}^{d \times u} \mid \Bbbk^\dagger \Bbbk = \mathbf{I}_u \}.
$
Here, $ \Bbbk $ is a $ \mathsf{k}d \times u $ matrix, corresponding to $\mathsf{k}$ Kraus operators in our case.
The tangent space $ T_\Bbbk \mathcal{M} $ at any point $ \Bbbk $ consists of all matrices $ \mathbf{V} \in \mathbb{C}^{d \times u} $ that satisfy  $
\Bbbk^\dagger \mathbf{V} + \mathbf{V}^\dagger \Bbbk = 0.
$
This ensures that infinitesimal small perturbations of $ \Bbbk $ remain on the Stiefel manifold.

For a curve $ \gamma(\tau) $ on the manifold, the tangent vector is defined through the gradient
$
\mathbf{V} = \left. \frac{d\gamma(\tau)}{d\tau} \right|_{\tau=0}.
$ 
Since $ \mathbf{V} $ represents the gradient of our cost function, we write it as
$
\mathbf{V} = \rm{grad} C(\Bbbk).
$
To ensure differentiability, we use the Riemannian manifold $ (\mathcal{M}, g) $, where $ g $ is a smoothly varying inner product on the tangent space
$
g_\Bbbk(\mathbf{V}, \mathbf{W}) = \rm{Tr}(\mathbf{V}^\dagger \mathbf{W}).
$
This leads to the geodesic equation  
$
\gamma(t) = \rm{Exp}_\Bbbk(\mathbf{V}).
$

For optimization, we use gradient-based methods. Given a function $ C: \mathcal{M}(d,u) \to \mathbb{R}$, the Riemannian gradient is obtained by projecting the Euclidean gradient onto the tangent space
\begin{equation}
\rm{grad} C(\Bbbk) = \nabla C(\Bbbk) - \Bbbk \rm{Sym}(\Bbbk^\dagger \nabla C(\Bbbk)),
\end{equation}  
where $ \rm{Sym}(\mathbf{A}) = (\mathbf{A}+\mathbf{A}^\dagger)/2 $.
Using Riemannian gradient descent, we update $ \Bbbk $ iteratively:  
\begin{equation}
\Bbbk_{t+1} = \rm{Retract}_\Bbbk(-\alpha \rm{grad} C(\Bbbk)),
\end{equation}  
where $ \alpha $ is the learning rate, and the retraction ensures the update on the manifold.

A retraction is a smooth mapping that approximates the exponential map while ensuring that updates remain on the manifold  
$
\rm{Retract}_\Bbbk: T_\Bbbk\mathcal{M} \to \mathcal{M}.
$
It satisfies $ \rm{Retract}_\Bbbk(0) = \Bbbk $ and approximates the true geodesic update 
$
\rm{Retract}_\Bbbk(\mathbf{V}) = \exp_\Bbbk(\mathbf{V}) + O(||\mathbf{V}||^2).
$ 
This preserves the constraint $ \Bbbk^\dagger\Bbbk = \mathbf{I}_u$.
Common retraction methods include QR-based methods, polar decomposition, the Cayley transform, and the exponential map. 
The QR-based retractions use a first-order approximation, whereas polar and Cayley methods are second-order, and the exponential map is the exact one. See Tab.~\ref{tab:retractions} for more detail. In this work, we apply the Cayley method. 
Finally, during optimization, $ \Bbbk $ moves along the geodesic $ \gamma(\tau) $ toward the optimal point representing $ \mathcal{E} $. See Figure.~\ref{fig:12} for an illustration.

\medskip
\textbf{Acknowledgments} \par 
This paper is supported by JSPS KAKENHI Grant No.23K13025.

\medskip
\textbf{Conflict of interest}  \par
The authors declare no conflict of interest.

\medskip
\textbf{Data Availability Statement} \par
The data that support the finding of this study are available from the corresponding author upon reasonable request.

\bibliography{ref}

\end{document}